\documentclass{proposalnsf}

\usepackage[utf8]{inputenc}
\usepackage{amsfonts}
\usepackage{amsmath, amssymb,amsthm}
\usepackage{mathtools}
\usepackage{hyperref}
\usepackage{tcolorbox}
\usepackage{enumitem}
\usepackage{titlesec}
\usepackage{xcolor}
\usepackage{tikz-cd}
\usepackage{float}
\usepackage{algorithm}
\usepackage{algpseudocode}
\usepackage{braket}
\usepackage{graphicx}
\usepackage{caption}

\usepackage{tikz}
\usetikzlibrary{positioning, arrows.meta, shapes.geometric}
\usetikzlibrary{decorations.pathreplacing,positioning,calc,backgrounds}

\newcommand{\mb}{\mathbb}

\newcommand{\mc}{\mathcal}

\newcommand{\mtt}{\mathtt}

\newcommand{\tx}{\text}

\newcommand{\wgt}{\tx{\normalfont wt}}

\newcommand{\HGP}{\tx{\normalfont HGP}}

\newcommand{\rank}{\tx{\normalfont rank}}

\newcommand{\obs}{\tx{\normalfont obs}}

\newcommand{\iref}[2]{(\hyperref[#2]{\ref*{#1}.\ref*{#2}})}

\newcommand{\RS}{\mathrm{RS}}

\newcommand{\SSP}{\mtt{SS}}
\newcommand{\FF}{\mathbb{F}}

\newcommand{\F}{\mathbb{F}}

\newcommand{\eps}{\epsilon}

\newcommand{\Hamm}{\mtt{H}}

\newcommand{\transpose}{\mtt{t}}
\newcommand{\Nbr}{\tx{\normalfont Nbr}}

\newcommand{\CC}{\mathbb{C}}
\newcommand{\calC}{\mathcal{C}}

\newcommand{\vs}{s}
\newcommand{\vE}{E}
\newcommand{\vx}{x}
\newcommand{\vy}{y}

\newcommand\nodot[1]{}

\def\final{0}  
\def\iflong{\iffalse}
\ifnum\final=0  

\else 
\fi

\newtheorem{theorem}{Theorem}[section]
\newtheorem{Definition}[theorem]{Definition}

\newtheorem{claim}[theorem]{Claim}

\newtheorem*{claim*}{Claim}

\newtheorem{proposition}[theorem]{Proposition}

\newtheorem{lemma}[theorem]{Lemma}
\newtheorem{corollary}[theorem]{Corollary}

\newtheorem{remark}[theorem]{Remark}

\newif\ifshowcorrections
\showcorrectionstrue   

\definecolor{correctioncolor}{RGB}{0,90,160}

\title{Noisy-Syndrome Decoding of Hypergraph Product Codes}
\date{}

\author{
Venkata Gandikota\thanks{Syracuse University, Syracuse, USA}
\and
Elena Grigorescu\thanks{University of Waterloo, Waterloo, Canada}
\and
Vatsal Jha\thanks{Purdue University, West Lafayette, USA}
\and
S. Venkitesh\thanks{Institute for the Theory of Computing, Ben Gurion University of the Negev, Beersheva, Israel}
}

\begin{document}

\maketitle

\begin{abstract}
    Hypergraph product codes  are a prototypical family of quantum codes with state-of-the-art decodability properties. 
    In this work we consider the \emph{noisy} syndrome decoding problem and  exact recovery problem for hypergraph product codes and show a reduction to the decoding and exact recovery of classical codes in the noisy syndrome setting. 

    Our results hold for a broad class of codes admitting efficient syndrome decoding, including Sipser-Spielman codes and Reed–Solomon codes.
\end{abstract}
\tableofcontents
\section{Introduction}

A central problem in quantum fault tolerance is the construction of quantum codes with large distance, efficient decoding algorithms, and simple structure. Following the seminal work of Tillich and Z\'{e}mor \cite{Tillich13}, a sequence of length-$N$ quantum codes with distance $N^{1/2+\Omega(1)}$ was proposed in \cite{hastings2021fiber, breuckmann2021balanced, hastings2023weight}. This was later improved to near-linear distance $\Theta(N/\log{N})$ by Panteleev and Kalachev \cite{panteleev2022quantum}, and subsequently to linear distance and linear dimension in \cite{panteleev2022asymptotically, leverrier2022quantum, dinur2023good}.

The \emph{hypergraph product} (HGP) codes introduced in \cite{Tillich13} have played a central role in this line of work and continue to yield state-of-the-art results. The HGP framework combines arbitrary linear codes $C_{1},C_{2}$ with good parameters to yield the $HGP(C_{1},C_{2})$ with good parameters.  Most relevant to our work, Golowich and Guruswami \cite{golowich2024decoding} showed that HGP codes constructed from an expander code $C_1$ and the repetition code $C_2$ can be decoded from a constant fraction of errors relative to the code distance. Their approach reduces the quantum decoding task to \emph{syndrome decoding} of the classical constituent codes $C_1$ and $C_2$, making the procedure entirely classical.

At the core of these algorithms lies a classical computation. The decoder is given a syndrome, obtained by performing a small number of low-weight parity-check measurements on the quantum state, and must compute a correction to apply to the physical qubits. 
 However, in practice the syndrome measurements themselves are performed on noisy hardware, so the reported syndrome bits can be wrong independently of whether any data error occurred. Hardware faults such as crosstalk between simultaneously measured ancillas, or qubit decoherence tend to flip the measured syndrome bits. Moreover, experiments on IBM superconducting devices have directly confirmed that repeated syndrome measurement cycles produce statistics inconsistent with standard independent noise models~\cite{gicev2024syndrome}, and that qubit leakage out of the computational subspace causes persistent faults across multiple syndrome rounds that standard decoders are not designed to handle~\cite{lacroix2024soft}. 
This gives rise to the \emph{noisy syndrome decoding} problem, where both the data and the syndrome are subject to errors. 

In this setting, two natural goals arise. The first, which we call \emph{stable decoding}, requires that the decoder's output degrade gracefully with syndrome noise: small corruption in the syndrome should lead to only small errors in the decoder's estimate. 
More formally, we say that a code $\mathcal{C}$ is $(t,\eta,\alpha)-$\emph{stable noisy syndrome decodable} if there exists a decoder $\mathcal{D}$ 
that for any data error vector $e$ and a syndrome error vector $z$ with $\wgt_\Hamm(e)\leq t$ and $\wgt_\Hamm(z)\leq \eta$, 
outputs an estimate $\mathcal{D}(He + z) = \hat{e}$ such that 
$\wgt_\Hamm(e-\hat{e})\leq \alpha\cdot \wgt_\Hamm(z).$ Here, for $w\in \F_q^n$, $\wgt_\Hamm(w)=|\{i| w_i\ne 0\}|$.
The construction of stable noisy syndrome decoding was previously considered by \cite{spielman1995linear} using the 
expander codes of \cite{SipserSpielman1996expander}. 
This notion of stable noisy syndrome decoding extends naturally to quantum CSS codes, as studied in \cite{reshape, leverrier2015quantumexpander, fawzi2020constant}. We defer the formal definition to Definition \ref{def:quantum-nsd}.

The second goal, which we call \emph{exact recovery}, is more stringent, requiring the decoder to recover the true error exactly even in the presence of syndrome noise, i.e., $\wgt_\Hamm(e-\hat{e})=0$. 
We provide formal definitions for both notions, in the classical and quantum settings, in Section~\ref{sec:noisy-syndrome-defs}.

\subsection{Our results}
We show that if the classical constituent code is \emph{nice enough}, then the resulting HGP code inherits efficient noisy-syndrome decoding and exact recovery via a reduction to the classical problem. We make this precise in two complementary theorems.

\begin{theorem}[Informal, Stable decoding version]\label{thm:main-informal-nonnoisy}
	Let \(\FF_q\) be a finite field of characteristic \(2\).
	Let \(C\) be an explicit \(\FF_q\)-linear code with parameters \([n, \Theta(n), \Theta(n)]\).
	Suppose \(C\) admits $(\Theta(n),\Theta(n),O(1))$-\emph{stable noisy syndrome decoder} that runs in $T(n)$ time. 

	Then there exists an explicit hypergraph product code
	which admits \emph{stable noisy syndrome decoding} from
	\(t=\Theta(n)\) data errors and
	\(\eta=\Theta(n)\) syndrome errors, with $\alpha=1/2$
	and total decoding time \(O(nT(n)+n^{3})\).
\end{theorem}

\noindent We instantiate \ref{thm:main-informal-nonnoisy} with HGP codes obtained from the expander codes of~\cite{SipserSpielman1996expander} -- see  Corollary \ref{cor: ss-instatiation} and Corollary \ref{cor:zemor-syndrome-decoder}. Indeed, the decoders in~\cite{SipserSpielman1996expander} can be viewed as noisy syndrome decoders.
For the sake of completeness, we explicitly describe these syndrome decoding variants
 for expander codes obtained from bipartite vertex expanders (see Algorithm \ref{alg:expander-decoding-flags}) and from spectral expanders (see Algorithm \ref{algo:syndrome-zemor}), respectively. Hence, we obtain explicit HGP codes from expander-based codes that admit \((\Theta(n),\Theta(n),1/2)\)-stable noisy syndrome decoders that run in time \(O(n^3)\) as $T(n)=O(n)$ for the aforementioned codes.

\begin{theorem}[Informal, Exact recovery version]\label{thm:main-informal-noisy}
	Let \(\FF_q\) be a finite field of characteristic \(2\).
	Let \(C \subseteq \FF_q^{n}\) be an explicit linear code such that
	both \(C\) and its dual \(C^\perp\) have parameters $[n,\Theta(n),\Theta(n)]$, and they both have  syndrome decoders that
	correct a linear number of syndrome and data errors
	in time \(T(n)\).

	Then there exists a hypergraph product code that can be corrected from $O(n)$ adversarial data errors even in the presence of $O(n)$ syndrome errors in time \(O(nT(n)+n^{3})\).
\end{theorem}

\noindent We instantiate  Theorem \ref{thm:main-informal-noisy} with HGPs obtained from Reed-Solomon (RS) codes (see Corollary \ref{cor:RS}).  Syndrome decodability for RS codes up to half the minimum distance is precisely~\cite[Theorem 2 and Algorithm 2]{SSB07}.

\begin{remark}
We emphasize that the stabilizer checks of the quantum CSS codes we construct have weight $O(\sqrt{N})$, where $N$ is the code length; that is, the check weights are sublinear but not necessarily constant. An interesting open question is whether our techniques can be extended to yield quantum CSS codes with $O(1)$ weight stabilizer checks that still admit efficient noisy syndrome decoding and recovery.    
\end{remark}

\subsection{Noisy Syndrome Decoding: Formal Definitions}\label{sec:noisy-syndrome-defs}

We now make the notions of stable decoding and exact correction precise.

\subsubsection*{Classical setting}

We begin with stable noisy syndrome decoding, which captures the requirement that a small amount of syndrome noise should cause only a proportionally small error in the decoder's output.

\begin{Definition}[Stable Noisy Syndrome Decoding (Classical)]\label{def:classical-nsd}
	Let $C \subseteq \FF_q^n$ be a linear code with parity-check matrix
	$H \in \FF_q^{m\times n}$.
	We say that $C$ admits \emph{$(t,\eta,\alpha)$-stable noisy syndrome decoding}
	if there exists a decoding algorithm $\mathcal{D}$ such that for all data errors
	$e \in \FF_q^n$ and syndrome errors $z \in \FF_q^m$ satisfying
	\[
	\wgt_\Hamm(e) \le t \quad \text{and} \quad \wgt_\Hamm(z) \le \eta,
	\]
	given access to the corrupted syndrome
	\[
	\tilde{s} = He + z,
	\]
	the decoder outputs an estimate $\hat e = \mathcal{D}(\tilde{s})$ satisfying
	\[
	\wgt_\Hamm(\hat e - e) \le \alpha \cdot \wgt_\Hamm(z).
	\]
\end{Definition}

\begin{remark}\label{lem:stability}
	Two aspects of this definition are worth highlighting.
	\begin{itemize}
		\item When $z=0$, the definition requires exact recovery of $e$.
		\item When syndrome noise is present, the guarantee is one of \emph{stability}: the decoder's output deviates from the true error by at most a multiplicative factor of the syndrome noise. In particular, if decoding is repeated with fresh syndrome noise $z'$, the accumulated residual error remains $O(\wgt_\Hamm(z)+\wgt_\Hamm(z'))$.
	\end{itemize}
\end{remark}

The stable decoding guarantee allows the decoder's output to be slightly wrong when the syndrome is noisy. A stronger requirement is that the decoder recover the data error \emph{exactly}, even in the presence of syndrome noise.

\begin{Definition}[Exact Recovery (Classical)]\label{def:classical-nsc}
	Let $C \subseteq \FF_q^n$ be a linear code with parity-check matrix
	$H \in \FF_q^{m\times n}$.
	We say that $C$ admits \emph{$(t,\eta)$-exact recovery}
	if there exists a decoding algorithm $\mathcal{D}$ such that for all data errors
	$e \in \FF_q^n$ and syndrome errors $z \in \FF_q^m$ satisfying
	\[
	\wgt_\Hamm(e) \le t \quad \text{and} \quad \wgt_\Hamm(z) \le \eta,
	\]
	given access to the corrupted syndrome
	\[
	\tilde{s} = He + z,
	\]
	the decoder outputs the exact error vector $e$.
\end{Definition}

\subsubsection*{Quantum (CSS) setting}

Both notions extend naturally to CSS codes, which are built from a pair of classical codes $C_{X}$ and $C_{Z}$ satisfying $C_{X}^{\perp}\subseteq C_{Z}$. There are two differences from the classical case. First, the decoder must handle corrupted $X$- and $Z$-syndromes simultaneously. Second, correctness is measured up to \emph{stabilizer equivalence}. 
A CSS code has a distinguished set of operators, called stabilizers, that act trivially on the code space, so two errors that differ by a stabilizer are physically indistinguishable and must be treated as equivalent. 
Therefore, we measure errors using the \emph{stabilizer-reduced weight}
\[
|e_{\alpha}|_R \coloneqq \min_{c \in C_{\bar{\alpha}}} \wgt_\Hamm(e_{\alpha}+c),
\]
where $\alpha\in \{X,Z\}$ and $\bar{\alpha}=Z$ if $\alpha=X$ and $\bar{\alpha}=X$ if $\alpha=Z$.
\begin{Definition}[Stable Noisy Syndrome Decoding (Quantum, CSS)]\label{def:quantum-nsd}
	Let $C = \mathrm{CSS}(C_X,C_Z)$ be a CSS quantum code with parity-check matrices
	$H_X$ and $H_Z$.
	We say that $C$ admits \emph{$(t,\eta,\alpha)$-stable noisy syndrome decoding}
	if there exists a decoding algorithm $\mathcal{D}ec$ such that for all Pauli errors
	$e=(e_X,e_Z)$ and syndrome errors $D=(D_X,D_Z)$ satisfying
	\[
	|e_X|_R \le t,\quad |e_Z|_R \le t,
	\qquad\text{and}
	\qquad \wgt_\Hamm(D_X) \le \eta,\quad \wgt_\Hamm(D_Z) \le \eta,
	\]
	given access to the corrupted $X$- and $Z$-syndromes
	\[
	\tilde{s}_X = H_X e_Z + D_X,
	\qquad
	\tilde{s}_Z = H_Z e_X + D_Z,
	\]
	the decoder outputs corrections \((\hat e_X,\hat e_Z) = \mathcal{D}ec(\tilde{s}_X,\tilde{s}_Z)\) satisfying
	\[
	|e_X + \hat e_X|_R \le \alpha \cdot \wgt_\Hamm(D_Z),
	\qquad
	|e_Z + \hat e_Z|_R \le \alpha \cdot \wgt_\Hamm(D_X).
	\]
\end{Definition}

\begin{Definition}[Exact Recovery (Quantum, CSS)]\label{def:quantum-nsc}
	Let $C = \mathrm{CSS}(C_X,C_Z)$ be a CSS quantum code with parity-check matrices
	$H_X$ and $H_Z$.
	We say that $C$ admits \emph{$(t,\eta)$-exact recovery}
	if there exists a decoding algorithm $\mathcal{D}ec$ such that for all Pauli errors
	$e=(e_X,e_Z)$ and syndrome errors $D=(D_X,D_Z)$ satisfying
	\[
	|e_X|_R \le t,\quad |e_Z|_R \le t,
	\qquad\text{and}
	\qquad \wgt_\Hamm(D_X) \le \eta,\quad \wgt_\Hamm(D_Z) \le \eta,
	\]
	given access to the corrupted $X$- and $Z$-syndromes
	\[
	\tilde{s}_X = H_X e_Z + D_X,
	\qquad
	\tilde{s}_Z = H_Z e_X + D_Z,
	\]
	the decoder outputs corrections \((\hat e_X,\hat e_Z) = \mathcal{D}ec(\tilde{s}_X,\tilde{s}_Z)\) such that:
	\[\hat{e}_{X}=\arg \min_{c \in C_{Z}} \wgt_\Hamm(e_{X}+c),\; \qquad \hat{e}_{Z}=\arg \min_{c \in C_{X}} \wgt_\Hamm(e_{Z}+c)\]
\end{Definition}
\subsection{Proof Overview}

As discussed above, our results address two aspects of decoding in the presence of syndrome noise: \emph{error reduction} (stable decoding) and \emph{error correction} (exact recovery). We briefly sketch the ideas behind each.

\paragraph*{The starting point: splitting the noisy syndrome equation.}
The generic noisy syndrome equation for a hypergraph product code built from parity-check matrices $H_1$ and $H_2$ takes the form
\begin{equation}\label{eq:hgp-syndrome}
\vs + \vE = (H_1 \otimes I)\, \vx + (I \otimes H_2^\transpose)\, \vy,
\end{equation}
where $x$ and $y$ are the data errors, $E$ is the syndrome error, and we suppress dimensions for readability.

The key idea is to absorb the first term into the syndrome noise by writing $\vE' \coloneqq \vE + (H_1 \otimes I)\,\vx$, so that \eqref{eq:hgp-syndrome} becomes
\begin{equation}
\vs + \vE' = (I \otimes H_2^\transpose)\, \vy.
\end{equation}
Using the Kronecker product structure, this splits into independent blocks indexed by $i$:
\begin{equation}\label{eq:block-syndrome}
\vs^{(i)} + \vE'^{(i)} = H_2^\transpose\, \vy^{(i)},
\end{equation}
where the range of $i$ is determined by the dimensions of $H_1$. Each block is now a noisy syndrome equation for $H_2^\transpose$. If the classical code with parity check matrix $H_2^\transpose$ admits good noisy syndrome decoding or correction, we can recover (or approximate) $\vy^{(i)}$ and hence reconstruct $\vE'$. Unfolding the definition $\vE' = \vE + (H_1 \otimes I)\,\vx$ then yields a second noisy syndrome instance, this time for $H_1$, which we solve in a similar fashion.

\paragraph*{Stable noisy syndrome decoding via expander codes.}
For stable noisy syndrome decoding, we instantiate both $H_1$ and $H_2$ using Sipser--Spielman expander codes, whose classical noisy syndrome decoders run in near-linear time. The two-stage reduction sketched above then yields a quantum decoder that corrects $\Theta(\sqrt{N})$ errors in $\Theta(N)$ time, where $N$ is the length of the HGP code.

\paragraph*{Exact recovery via the augmented parity check matrix.} 
For exact recovery, we exploit the observation that the noisy syndrome equation for a parity-check matrix $H$ can be rewritten as
\begin{equation}
s = [I \;:\; H] \begin{bmatrix} \vE \\ \vx \end{bmatrix},
\end{equation}
which is simply a standard (noiseless) syndrome decoding problem for the augmented matrix $[I : H]$. If the code with parity check $[I : H]$ has favorable properties like large distance and an efficient syndrome decoder then one can correct the data error $\vx$ exactly, even in the presence of syndrome noise $\vE$. 
Polynomial codes (Reed--Solomon) as well as Sipser--Spielman codes give rise to parity-check matrices $[I : H]$ with precisely these properties, yielding exact recovery in time $\Theta(N^{3/2})$.

\paragraph*{Why error reduction does not imply error correction.}
It is worth pointing out that one cannot simply iterate the error reduction procedure to achieve exact correction. There are two fundamental obstacles. First, it is not guaranteed that the weight of the error strictly decreases after each iteration as it may plateau. Second, each new iteration requires a fresh syndrome computation, which may introduce additional noise. For these reasons, error correction requires genuinely different techniques.

\subsection{Related Work}

The problem of decoding hypergraph product (HGP) codes has been studied in a
number of works, which differ significantly from ours in either their noise models,
algorithmic approaches, or performance guarantees. As mentioned earlier our main contribution is the exact recovery in the presence of adversarial syndrome noise, which was not considered before.

\paragraph*{Quantum expander codes and small-set-flip}
Leverrier, Tillich, and Zémor~\cite{leverrier2015quantumexpander}
introduced quantum expander codes together with the small-set-flip decoder, providing the
first linear-time decoding algorithm for a family of HGP codes.
Their analysis handles \emph{adversarial} data errors of weight up to $\Theta(\sqrt{N})$ with noise-free syndromes.
Subsequent work of Fawzi, Grospellier, and Leverrier~\cite{fawzi2020constant} extended the analysis to stochastic data and syndrome errors, establishing threshold and fault-tolerance results showing that decoding succeeds with high probability below a certain noise rate.
 
In contrast, our work considers a \emph{fully adversarial noise model} for both
data errors and syndrome errors.
Rather than analyzing a specific quantum decoding algorithm such as
small-set-flip, we develop a general \emph{reduction-based framework} showing
that noisy-syndrome decodability and correction of HGP codes follows from corresponding
classical noisy-syndrome decoding guarantees of the underlying constituent codes.
Our results therefore apply beyond quantum expander codes and do not rely on
probabilistic noise assumptions.
For the adversarial noise model, \cite{fawzi2020constant} provides a stable noisy-syndrome decoder for quantum expander codes of length $N$, against $O(\sqrt{N})$ adversarial errors in the syndrome and with a $\Theta(N)$ running time, which is quite close to our runtime $O(N^{3/2})$ in Theorem~\ref{thm:main-informal-nonnoisy} when instantiated with expander codes. In our construction, we are able to \emph{exactly recover} the data errors even in the presence of $O(\sqrt{N})$ adversarial syndrome errors with time complexity $\Theta(N^{3/2})$.

\paragraph*{ReShape decoder}
The ReShape decoder~\cite{reshape} provides a general adversarial decoder for
arbitrary $[[N,K,D]]$ hypergraph product codes, assuming access to a classical decoder for the underlying code.
The algorithm makes $O(K)$ calls to the classical decoder and incurs an
additional $O(N^2)$ overhead (Theorem 1 in \cite{reshape}).
In the typical regime where $K = \Theta(N)$, this yields $O(N)$ oracle calls
and an overall running time of $O(N^2)$ (or $O(N)\cdot T$, where $T$ is the running time of the classical decoder).
For instance, when instantiated with Sipser--Spielman expander codes, where
$T = O(\sqrt{N})$, the total runtime remains superlinear.

Crucially, the ReShape decoder assumes \emph{perfect (noiseless) syndrome
information} and does not provide guarantees in the presence of syndrome noise; indeed, its
correctness can fail even under a small number of corrupted syndrome bits.

By contrast, the decoding framework developed in this paper explicitly handles
\emph{noisy syndromes} in an adversarial setting.
Our reductions yield quantum decoders with running time
$O(N^{3/2})$, matching the optimal scaling up to subpolynomial factors.
Our work can thus be viewed as extending ReShape-style reductions
to the noisy-syndrome regime while simultaneously improving the asymptotic
runtime.

\paragraph*{Expander-based and Viderman-style decoders}
In~\cite{wootters2023viderman}, a quantum analogue of Viderman's classical
expander decoding algorithm was proposed, reducing quantum decoding to a sequence of erasure decoding
problems and yielding an adversarial decoder for HGP codes with running time
$O(N^{3/2})$.
Similarly, Golowich and Guruswami~\cite{golowich2024decoding} reduced decoding of
HGP codes to noisy-syndrome decoding of classical codes derived from one-sided
vertex expanders, again achieving $O(N^{3/2})$ running time.

Our results differ in several key respects.
First, we provide a unified reduction framework that applies to both noisy and
noiseless syndrome settings and isolates \emph{stable noisy syndrome decoding}
as the fundamental classical primitive.

Second, the HGP codes considered in~\cite{golowich2024decoding} require that both $H_{*}$ and $H^{T}_{*}$ define classically noisy-syndrome decodable codes (for $*\in\{X,\, Z\}$), whereas the constructions in the present paper require noisy syndrome decodability only for $H_{*}$.

\paragraph*{Summary}
In summary, prior works either:
(i)~achieve linear-time decoding under stochastic noise assumptions~\cite{leverrier2015quantumexpander},
(ii)~provide general adversarial decoders assuming noiseless syndromes~\cite{reshape}, or
(iii)~handle adversarial noise with superlinear runtime~\cite{wootters2023viderman,golowich2024decoding}.
The present work complements these results by giving a general,
reduction-based approach to decoding and recovering from \emph{adversarially noisy
syndromes} in near-linear time, thereby addressing an open question posed
in~\cite{golowich2024decoding}.

\subsection{Open Problems}
We highlight a few questions that we find particularly compelling.

\begin{enumerate}
	\item \textbf{Tolerating more syndrome noise.}
	The syndrome error tolerance of our quantum decoders is ultimately bottlenecked
	by the noisy syndrome decoding capabilities of the underlying classical codes.
	A natural goal is to push this tolerance further.
	For polynomial codes such as Reed--Solomon codes, 
	the folding operation (to obtain folded Reed--Solomon codes) is known to dramatically improve the decoding radius in the curve-fitting and list decoding settings.
	Does a similar phenomenon hold for syndrome decoding?
	If so, this would directly translate into more robust hypergraph product codes via our reductions.

	\item \textbf{New code families.}
	Our main theorems reduce quantum noisy syndrome decoding of hypergraph
	product codes to classical noisy syndrome decoding of the constituent codes.
	This puts the spotlight squarely on the classical side: which other code
	families satisfy the assumptions of
	Theorems~\ref{thm:main-informal-noisy} and~\ref{thm:main-informal-nonnoisy}?
	For polynomial-based codes, a judicious choice of evaluation points seems
	to be the key ingredient. We wonder whether similar ideas can be made to
	work for multiplicity codes, Reed--Muller codes, or other algebraic and
	combinatorial families, while still ensuring efficient syndrome decoding for both the code and its dual.

	\item \textbf{Random noise.}
	Our guarantees are worst-case, handling fully adversarial data and syndrome
	errors.
	It is natural to ask what happens under random noise.
	Can our reduction framework be married with the probabilistic techniques
	used for quantum expander codes to obtain better decoding thresholds or
	average-case guarantees in the noisy-syndrome setting?

\end{enumerate}
\section{Preliminaries}

In this section, we will recall some essential preliminaries on quantum codes, and cover a few more in Appendix~\ref{app:additional-prelims}.

\paragraph*{Convention.}  In this work, we will assume that \(\FF_q\) is a finite field having characteristic 2, that is, \(q=2^r\) for some \(r\ge1\).  We will also assume that all our classical codes are linear over \(\FF_q\).  We will stick to these conventions throughout, without further mention.

\subsection{Classical error-correcting codes}

Let us briefly recall the important families of classical codes that we will consider in this work.

\paragraph*{The Sipser-Spielman expander code.}  Let \(\FF_q\) be a finite field, and \(G\) be a \((d_\ell,d_r)\)-biregular bipartite graph with bipartition \([n]\sqcup[d_\ell n/d_r]\). 
Let \(C\subseteq\mb{F}_q^{d_r}\) be a linear code.  
By interpreting the left set of vertices as the $n$ locations for codeword entries, and the right vertices as indexing $d_\ell n/d_r$ many constraints satisfied by the codeword entries,~\cite{SipserSpielman1996expander} defined a code
\[
\SSP(G,C)=\{y\in\FF_q^n:y(\Nbr(j))\in C\tx{ for all }j\in[d_\ell n/d_r]\},
\]
where we denote \(y(\Nbr(j))=(y_t:t\in\Nbr(j))\in\mb{F}_q^{d_r}\).  This code is now called the \emph{Sipser-Spielman code}.

 Further, for \(\gamma>0,\; \epsilon\in(0,1/2)\), we say the bipartite graph \(G\) for a $(d_{l},\; d_{r})$-biregular bipartite graph as defined above is a \emph{\((\gamma,d_{l}(1-\epsilon))\)-vertex expander} if we have \(|\Nbr(S)|\ge d_{l}(1-\epsilon)|S|\), for all \(S\subseteq[n]\) with \(|S|\le\gamma n\).   The \emph{Sipser-Spielman expander code} is the code \(C_G\) when \(G\) is an expander. The work \cite{SipserSpielman1996expander} showed that the Sipser-Spielman expander code \(C_G\) (with a random expander graph) will have distance $\geq \gamma n$ and  can correct $\le \gamma(1-2\epsilon)n$ errors in linear time when $\epsilon\in (0,1/4)$. Explicit expander graphs with the required parameters have since been constructed, for instance, by~\cite{golowich-lossless}.

\paragraph*{Reed-Solomon codes and variants.}  Let \(\gamma\in\F_q^\times\) have multiplicative order \(n\).  Then the \emph{Reed-Solomon (RS) code} with evaluation points \(\{1,\gamma,\ldots,\gamma^{n-1}\}\) is given by
\[
\RS_q(\gamma;n;k)=\left\{[f]=(f(1),f(\gamma),\ldots,f(\gamma^{n-1})):f(X)\in\F_q[X],\,\deg(f)<k\right\}.
\]
It follows by definition, and basic properties of polynomials, that the code \(\RS_q(\gamma;n;k)\) has dimension \(k\) and distance \(n-k+1\).

Of course, the Reed--Solomon codes are defined more generally for an arbitrary choice of evaluation points, but we will only consider the above specific choice of evaluation points.

\subsection{Quantum error-correcting codes}

We now discuss the essentials of hypergraph product codes.  More on the essentials of quantum codes are covered in Appendix~\ref{app:additional-prelims}.

\paragraph*{HGP codes.}  The \emph{hypergraph product code construction} can combine \emph{any} two classical linear codes (of \emph{any} lengths) in a way that directly yields a CSS code \cite{Tillich13}.
Specifically, if $H_1$ and $H_2$ are the parity check matrices for $C_1$ and $C_2$ and $H_{1}^\transpose$, $H_{2}^\transpose$ is the parity-check matrices for the transpose codes $C_{1}^\transpose$ and $C_{2}^\transpose$ respectively, then the \emph{hypergraph product code} $\HGP(H_1,H_2)$ is the CSS code with quantum parity check matrices $H_1'$ and $H_2'$ defined by
\[H_1' \coloneqq\begin{bmatrix}H_1 \otimes I & I \otimes H_2^\transpose\end{bmatrix}, \quad H_2' \coloneqq\begin{bmatrix}I \otimes H_2 & H_1^\transpose \otimes I\end{bmatrix}.\]
It is straightforward to verify that $H_1'(H_2')^\transpose=0$, and hence, this defines a valid CSS code.

Moreover, if $H_1$ and $H_2$ have dimensions $m_{1}\times n_{1}$ and $m_{2}\times n_{2}$ and $C_1$ and $C_2$ have parameters $[n_1,k_1,d_1]$ and $[n_2,k_2,d_2]$ and $C_{1}^\transpose$ and $C_{2}^\transpose$ have parameters $[m_{1},k_{1}^\transpose,d_{1}^\transpose]$ and $[m_{2},k_{2}^\transpose,d_{2}^\transpose]$ respectively, then \cite{Tillich13} show that $\HGP(H_1,H_2)$ is a code with parameters
\[ [[n_1n_2+m_{1}m_{2},k_1k_2+k_{1}^\transpose k_{2}^\transpose, \min\{d_1,d_2,d_{1}^\transpose,d_{2}^\transpose\}]], \]
where $k_{1}^\transpose,k_{2}^\transpose$, and  $d_{1}^\transpose,d_{2}^\transpose$ are the ranks and distances of the the codes with $H_{1}^\transpose$ and $H_{2}^\transpose$ as their respective parity check matrices.

\subsection{Some Helpful Claims}

We include a few easy claims that will be used repeatedly through our proofs. Their proofs appear in Appendix~\ref{app:helpful}.

For a binary code $C$ with parameters $[n,k,d]$ and parity-check matrix $H$ with dimension $(n-k)\times n$, one can obtain a self-orthogonal code, as given by the following claim:
\begin{claim}\label{claim:self-orthogonal}
    The code $\pi(C):=\{(c,c):c\in C\}$ has parameters $[2n,k,2d]$ and is self-orthogonal. Further, the matrix:
    \[H':=\begin{bmatrix}
    H & O_{(n-k)\times n}\\
    O_{(n-k)\times n} & H\\
    A & B
    \end{bmatrix}\] is a parity-check matrix for $\pi(C)$, for some matrix $[A:B]$ with $rank([A:B])=k$ .
\end{claim}

\begin{claim}\label{clm:symmetrization}
Let $[I_{n-k}:P_{(n-k)\times k}]$ and $[P^{T}_{k\times (n-k)}:I_{k}]$ be the systematic form of a parity-check matrix and the corresponding generator matrix for a binary linear code $C$ with parameters $[n,k,d]$. Consider the matrix $H':=\begin{bmatrix}
    I & P\\
    P^{t} & I
\end{bmatrix}.$ Then $Ker(H')=C'\cap {C^{'}}^{\perp},$ where $C':=Ker([I:P])$.
\end{claim}

\begin{claim}\label{systematic-rank}
    Let $H$ be a parity-check matrix for a code $C$ with parameters, $[n,k,d]$ and let $[I_{n-k}:P_{(n-k)\times k}]$ be its systematic form representation with $C':=Ker(P)$. If $rank(P)<k$ then $d(C')\geq d$.  
\end{claim}

To simplify our analysis for the HGP codes constructed, we consider parity-check matrices with a particular structure. Let $C$ be a $[n,k,d]$ code with the systematic form of the parity-check matrix being $[I_{n-k}:P_{(n-k)\times k}]$.

  In our proofs we will be using HGP codes, where it will be important to work with symmetric matrices. Furthermore, it will also allow us to reduce the syndrome decoding of HGP codes considered to that of the underlying classical codes. It turns out that considering the symmetric matrices $H'=\begin{bmatrix}
			O & P\\
			P^\transpose & O
		\end{bmatrix}$ will greatly simplify our proofs.

\begin{claim}\label{clm:distance-lem}
Let \(C\) be the $[n,k,d]$ code with parity check matrix \(H\).  Then the code defined by the parity check matrix \(H'\) has distance at least \(\min\{d(C),d(C^\perp)\}\).
\end{claim}

    \begin{claim}\label{clm:error-weight}
    Let $H\in\mathbb{F}_{q}^{m\times n}$ be a parity-check matrix and let $x=(x^{1},x^{2}, \ldots,x^{n})\in\mathbb{F}_{q}^{n^{2}}$ where $x^{i}\in\mathbb{F}_{q}^{n}$. If    $E_{m\cdot n \times 1}:=(H\otimes I_{n\times n})x$,
    then    $\wgt_\Hamm(E^{i})\leq \wgt_\Hamm(x)$,
    where 
    $$E=\begin{bmatrix}
        E^{1}_{n\times 1}\\
      \vdots\\
        E^{m}_{n\times 1}
    \end{bmatrix} \text{ and } E^{i}\in\mathbb{F}_{q}^{n\times 1} \text{ for } 1\leq i\leq m.$$
\end{claim}

We now reduce the noisy syndrome exact recovery problem to that of usual syndrome decoding in the absence of noise. We use the following  observation repeatedly:
\begin{claim}\label{noiseless syndrome reduction}
    Let $H\in\mathbb{F}_{2}^{(n-k)\times n}$  be the parity-check for a code $C$ and let $[I_{n-k}:P_{(n-k)\times k}]$ be its systematic form. Consider the matrix $H':=\begin{bmatrix}
        P_{(n-k)\times k} &O_{(n-k)\times k}\\
        O_{(n-k)\times k} & P_{(n-k)\times k}
    \end{bmatrix}.$    Then for an adversarial error, $E\in\mathbb{F}_{2}^{2(n-k)},$ and a data error $a\in\mathbb{F}_{2}^{2k}$ the equation for noisy syndrome exact recovery given by $s^{obs}+E=H^{'}a,$    can be reduced to noiseless syndrome decoding of $C$.
\end{claim}

\section{Exact Recovery for HGP codes}


In this section, we will prove the formal version of Theorem~\ref{thm:main-informal-noisy}. First, we formally restate the theorem statement.

	\begin{theorem}[Formal, Exact Recovery]\label{thm:main-formal-nonnoisy}
		Let \(\FF_q\) be a finite field of characteristic \(2\).
		Let \(C \subseteq \FF_q^{n}\) be an explicit linear code with parity check matrix $H=\begin{bmatrix}I_{n-k}&P_{(n-k) \times k}\end{bmatrix}$. 
        Further assume that both \(C\) and its dual \(C^\perp\) have parameters $[n,\Theta(n),\Theta(n)]$, and they both 
        admit \emph{exact} syndrome decoding
		from \(\Theta(n)\) errors in time \(T(n)\).

		Let
		$
		H'=\begin{bmatrix}
			O & P\\
			P^\transpose & O
		\end{bmatrix}
		$
		be defined from the systematic parity-check matrix $H$ of \(C\).
		Then the hypergraph product code \(\HGP(H',H')\) has parameters
		$
		[[\Theta(n^2),\Theta(n^2),\Theta(n)]].
		$

		Moreover, \(\HGP(H',H')\) admits $(t=O(n), \eta=O(n))$-\emph{exact recovery}
		for any Pauli error and syndrome noise of total weight \(O(n)\), in time \(O(nT(n))\).
	\end{theorem}

Using Claim~\ref{clm:distance-lem} we have the following proof for Theorem~\ref{thm:main-formal-nonnoisy}.
\begin{proof}[Proof of Theorem~\ref{thm:main-formal-nonnoisy}]
    Suppose \(C\) has parameters \([n,k,d]\).  So we have
    \[
    H=\begin{bmatrix}I_{n-k}&P_{(n-k)\times k}\end{bmatrix},\quad\tx{and}\quad H'=\begin{bmatrix}O_{(n-k)\times (n-k)} & P_{(n-k)\times k}\\
    P^\transpose_{k\times(n-k)} & O_{k\times k}
    \end{bmatrix}.
    \]
Note that since $H'$ is a symmetric $n \times n$ matrix, and $k = \Theta(n)$, it follows that the HGP construction, $\HGP(H', H')$ gives a code with parameters $[[\Theta(n^2),\Theta(n^2),\Theta(n)]]$. 

The noisy-syndrome equation for the aforementioned Hypergraph Product code can be then written as:
$$\vs_{\obs}+\vE=(H'\otimes I)\vx + (I\otimes H')\vy,$$
where $\vs_{\obs} \in \mb{F}_q^{n^2}$ is the observed (noisy) syndrome with error 
$\vE \in \mb{F}_q^{n^2}$ and $(\vx,\; \vy)\in \mb{F}_q^{2n^2}$ represent the $X$ ( or $Z$) data errors.

As the base field is of characteristic $2$, the above equation can be rephrased as
$$\vs_{\obs}+\vE+(H'\otimes I)\vx=(I\otimes H')\vy.$$

Substituting $E\coloneqq \vE + (H'\otimes I)\vx \in \mb{F}_q^{n^{2}}$ and by the Kronecker definition of tensor product, we can split the $n^{2}$ different syndrome equations into $n$ blocks of length $n$ each, giving: 
$$(\vs_{\obs})^{(i)}+E^{(i)}=H'\vy^{(i)} \qquad \text{ for } i \in [n],$$
where $E=(E^{1}, \ldots, E^{n})$ with $E^{i}\in\mathbb{F}_{q}^{n}$ and $\vs_{obs}=((\vs_{obs})^{(1)}, \ldots,(\vs_{obs})^{(n)})$ with $(\vs_{obs})^{(i)}\in\mathbb{F}_{q}^{n}$.

We now consider the following equivalent formulation of the above system of equations for $H'$.
$$(\vs_{\obs})^{(i)}=
\begin{bmatrix}I_{n\times n}&H'\end{bmatrix}
\begin{bmatrix}E^{(i)}\\ \vy ^{(i)}\end{bmatrix},$$
for $1\leq i\leq n.$
Substituting the value of $H'$, we get
$$(\vs_{\obs})^{(i)}=\begin{bmatrix}I_{(n-k)} &O _{(n-k)\times k}& : & O_{(n-k)\times (n-k)} & P\\
O_{k\times (n-k)} & I_{k\times k}& : &   P^\transpose & O_{k\times k}
\end{bmatrix}\begin{bmatrix}
    E^{(i)}\\
    \vy^{(i)}
\end{bmatrix}.$$

It is not difficult to see that the above system of equation can be rewritten as
$$(\vs'_{\obs})^{(i)}=
\begin{bmatrix} I&P \end{bmatrix}
\begin{bmatrix}{(E')}^{(i)}\\ {(\vy')}^{(i)}\end{bmatrix},
\quad\tx{where}\quad 
{(E')}^{(i)}\coloneqq\begin{bmatrix}{E_{1}^{(i)}}\\ E_{2}^{(i)}\\ \vdots \\E_{(n-k)}^{(i)}\end{bmatrix},\quad\tx{and}\quad
{(\vy')}^{(i)}\coloneqq\begin{bmatrix}\vy^{(i)}_{n-k+1}\\ \vdots \\ \vy^{(i)}_{n} \end{bmatrix}.$$
and,
$$({\vs''}_{\obs})^{(i)}=
\begin{bmatrix}
    P^\transpose&I
\end{bmatrix}
\begin{bmatrix}{(\vy'')}^{(i)}\\ {(E'')}^{(i)}
\end{bmatrix},
\quad\tx{where}\quad 
{(E'')}^{(i)}=\begin{bmatrix}
    E^{(i)}_{n-k+1}\\
    \vdots\\
    E^{(i)}_{n}
\end{bmatrix}\quad\tx{and}\quad 
{(\vy'')}^{(i)}=\begin{bmatrix}
    \vy^{(i)}_{1}\\
    \vdots\\
    \vy^{(i)}_{n-k}
\end{bmatrix}.$$
The two aforementioned equations correspond to the syndrome decoding for \(C\) and \(C^\perp\) respectively and hence using the fact that $C$ and $C^{\perp}$ have syndrome decoders running in time $T(n)$, we can recover both $E$, and $\vy$. 

Now, since $E = \vE + (H' \otimes I)\vx$, we obtain a similar reduction to $n$ instances of classical syndrome decoding for $C$ and $C^\perp$, which can be solved to recover $\vx$  and $\vE$ exactly in total time $O(nT(n))$. 
\end{proof}

\subsection{Instantiation with polynomial codes}
We now instantiate Theorem~\ref{thm:main-formal-nonnoisy} with Reed--Solomon codes that are known to posses syndrome decoders that can correct from noisy syndromes. 

\begin{corollary}\label{cor:RS}
    Let $\RS_q(\gamma, n,\Theta(n))$ be the Reed-Solomon code over a finite field of characteristic $2$, with evaluation points $\{1, \gamma, \ldots, \gamma^{n-1}\}$ for some $\gamma \in \F_q^{\times}$ of order $n$. Let $H = \begin{bmatrix}
        I&P
    \end{bmatrix}$ be its parity check matrix, and define by $H'= \begin{bmatrix}
        O & P\\
        P^\transpose & O
    \end{bmatrix}$.
    Then the HGP code $\HGP(H', H')$ is a quantum code with parameters $[[\Theta(n^2),\Theta(n^2),\Theta(n)]]$ that admits $(t=O(n), \eta=O(n))$-exact recovery for any Pauli error and syndrome error of total noise weight $O(n)$ in time $O(n^3)$.


\end{corollary}
\noindent The proof of Corollary~\ref{cor:RS} follows directly from Theorem~\ref{thm:main-formal-nonnoisy} and the syndrome decoder of \cite[Theorem 2 and Algorithm 2]{SSB07} that corrects  $\lfloor \frac{n-k}{2} \rfloor$ errors in time $O(n^2))$ for any $k = \Theta(n)$.

\section{HGP Codes from Vertex and Spectral Expanders}
In this section we provide two constructions of HGP codes 1) that are stable noisy syndrome decodable and 2) that admit exact recovery. Both these construction are derived from classical expander codes. 

\subsection{Stable Noisy Syndrome Decoding for HGP codes}
We first start with the stable noisy syndrome decodable HGP codes derived from Sipser-Spielman codes. We start by recalling the (classical) stable noisy syndrome decoding guarantees of   Sipser-Spielman~\cite{spielman1995linear} codes that will be crucially invoked by the HGP decoders constructed in this section. Note that \cite{spielman1995linear} had called them error-reduction codes, and we use the same terminology.

\begin{lemma}[Sipser--Spielman Stable Noisy Syndrome Decoding \cite{spielman1995linear}]
\label{lem:ss-error-reduction}
Let $B$ be a $(d,2d)-$ biregular biparite graph with $(\delta,\tfrac{3}{4}d+2)$ vertex expansion, and let $R(B)$ be the associated
error-reduction code with $n$ message bits and $n/2$ check bits.
There exists a linear-time algorithm with the following properties:

\begin{enumerate}
    \item[(a)] (\emph{Error reduction})  
    Given a word that differs from a codeword in $v$ message bits and $t$ check bits,
    where $v,t \le \delta n/2$, the algorithm outputs a word that differs from the
    codeword in at most
    $
        \max\{v/2,\; t/2\}
    $
    message bits.

    \item[(b)] (\emph{Error correction with clean checks})  
    If the input word differs from a codeword in at most $v \le \delta n/2$ message bits
    and in no check bits ($t=0$), then the algorithm outputs the original codeword.
\end{enumerate}
\end{lemma}
\begin{remark}
    In the original work~\cite{spielman1995linear}, the expander graph considered was a random construction, but we now know explicit expander graphs from several works.  We instantiate the specific choice of bipartite lossless expander in~\cite{golowich-lossless}, which allows \((\delta,(1-\eps)d)\)-vertex expansion for \(\eps\in(0,1/4)\) chosen so that \((1-\eps)d\ge\frac{3}{4}d+2\), and \(\delta=2^{-10}\).
\end{remark}

\subsubsection{The Construction}
Consider $H_{1}:=\begin{pmatrix}
    H & 0\\
    0& H
\end{pmatrix},$ where $H$ corresponds to a $(d_{L}, d_{R})$-biregular bipartite graph and is a $(\delta, d_{L}(1-\epsilon))$ vertex expander and the bipartite graph defined by the matrix $H^{t}$ is a $(\delta',d_{R}(1-\epsilon'))$ vertex expander where $\delta,\; \delta'>0$ and $\epsilon,\;  \epsilon'\in(0,1/4)$  respectively. Let $H':=[I_{n-k}:P_{(n-k)\times k}]$ be the systematic form of the parity-check matrix of a code satisfying the conditions in Theorem \ref{thm:main-formal-nonnoisy}. Consider the hypergraph product code $HGP(H_{1},H_{2})$ where $H_{2}=\begin{bmatrix}
    0 & P\\
      P^{t} & 0
\end{bmatrix}.$ This gives the corresponding noisy syndrome equations as:
\[s_{X}+E_{X}=(H_{1}\otimes I)a_{1}+(I\otimes H_{2}^{t})a_{2},\quad\text{and}\quad s_{Z}+E_{Z}=(I\otimes H^{t}_{1})b_{1}+(H_{2}\otimes I)b_{2},\]
where $a_{1},\; a_{2},\; b_{1},\; b_{2}$ are the data errors, $E_{X},\; E_{Z}$ are the syndrome errors and $s_{X},\; s_{Z}$ are the observed syndrome vectors.

\begin{theorem}
The hypergraph product code $HGP(H_{1},H_{2})$ is a $(N,\Theta(N),\Theta(\sqrt{N}))$ quantum code that admits a  $(\Theta(\sqrt{N}),\Theta(\sqrt{N}),1/2)$-stable noisy syndrome decoding in $\Theta(N^{3/2})$ time.

\end{theorem}
\begin{proof}
    We discuss for $X-$errors, the argument for $Z$-errors is symmetric.
    Consider the substitution $E'_{X}:=E_{X}+(H_{1}\otimes I)a_{1}$. This gives the equations: $s_{X}+E'_{X}=(I\otimes H_{2}^{t})a_{2}.$ These equations can be further simplified using the Kronecker product decomposition to,
    \[s^{(i)}_{X}+{E'}^{(i)}_{X}=H_{2}^{t}a_{2}^{(i)} \qquad \text{ for } i \in [n],\quad\text{where}\quad\beta^{(i)}=\begin{pmatrix}
        \beta^{(n\cdot(i-1)+1)}\\
       \vdots\\
        \beta^{(n\cdot i)}
        \end{pmatrix}\quad\text{for }\beta\in\{s_{X},E_{X}',a_{2}\}.\]


        Using Claim \ref{noiseless syndrome reduction} and by Theorem \ref{thm:main-formal-nonnoisy}, we are able to recover ${E'}_{X}^{(i)}$ for $\wgt_\Hamm(a_{1})$, $\wgt_\Hamm(a_{2})$, $\wgt_\Hamm(E_{X})=\Theta(n)$. This gives the system of equations:
        \[{E'_{X}}=E_{X}+(H_{1}\otimes I)a_{1},\]
        where $E^{'}_{X}$ is known.
        This is exactly the noisy-syndrome decoding problem for Sipser-Spielman codes to a permutation of coordinates, and hence we can apply the error-reduction algorithm for Sipser-Spielman codes as mentioned in \cite{spielman1995linear}.
\end{proof}
In the following section, we look at noisy syndrome exact-recovery for HGP codes obtained from vertex and spectral expanders.

\subsection{Exact Recovery for HGP Codes}

In this section we look at two constructions for obtaining HGP codes with noisy syndrome exact recovery. For the first construction we first consider expander codes obtained from bipartite vertex expanders and then use the distance balancing technique \cite{evra2022decodable}. For the second construction we are able to obtain codes with good $d_{X},\; d_{Z}$ right off the bat.
\subsubsection{Construction I}
Using the construction mentioned in this section, we can instantiate with Expander Codes. In particular, we need the noisy-syndrome correction for codes derived from  bipartite vertex expanders on biregular graphs. 

Let $H$ be an adjacency matrix for a $(c,d)-$ biregular $(\delta,c(1-\epsilon))$ bipartite vertex expander with $\epsilon\in(0,1/4)$. Let the systematic form representation for $H$ be $[I_{n-k}:P_{(n-k)\times k}]$. Now, consider the matrix:
\[H'=\begin{bmatrix}
    P & O_{(n-k)\times k}\\
    O_{(n-k)\times k} & P
\end{bmatrix}\]

Consider the hypergraph product code $HGP(H',{H'}^{t})$.


\begin{theorem}\label{thm:ldpc-dist}
    $HGP(H',{H'}^{t})$ is a quantum code with parameters$$[[N,\Theta(N), d_{X}=\Theta(\sqrt{N}), d_{Z}=O(1)]]$$ and has a decoder that runs in $O(N^{3/2})-$time and corrects  $\leq \frac{\delta}{2} (1-2\epsilon)n$ X-data errors and $
    \leq \frac\delta2 (1-2\epsilon)n$ X-syndrome errors.
\end{theorem}
\begin{proof}
The hypergraph product code $HGP(H',{H'}^{T})$ admits the following (noisy-) syndrome equations:
\[\vs^{obs}+\vE=(I\otimes H')\vx+(H'\otimes I)\vy\]

Taking $E:=\vE+(H'\otimes I)\vy$, we obtain:
\[s^{obs}+E=(I\otimes H')x\]
Now, by Claim \ref{noiseless syndrome reduction}, we have that every $(x,y)\in\mathbb{F}_{2}^{2n}$ and $E\in\mathbb{F}_{2}^{n}$ with $\wgt_\Hamm(x,y)\leq \delta(1-2\epsilon)n/2$ and $\wgt_\Hamm(E)\leq \delta(1-2\epsilon)n/2$ can be corrected for.

For $Z-$ errors we have a similar argument, but as the rows of $H$ are sparse with $O(1)$ weight, implying that the distance of $C^{\perp}=O(1)$, we can only correct for data errors and syndrome errors with constant weight.
\end{proof}
We can improve the $Z-$ error correction of the code by distance balancing as in \cite{evra2022decodable}.

\begin{proposition}[follows from Theorem 2.3 of \cite{evra2022decodable} and Theorem \ref{thm:ldpc-dist}]
   Given a quantum CSS code with parameters $[[N,\Theta(N),d_{X}=\Theta(\sqrt{N}),d_{Z}]]$ and a repetition code of length $l$ there exists a quantum CSS code that has parameters $[[\Theta(Nl),\Theta(N),d^{'}_{X}=\Theta{(\sqrt{N})},d^{'}_{Z}\geq ld_{Z}]]$
\end{proposition}

Now setting $l = \Theta(\sqrt{N})$ yields a CSS code of length $\Theta(N^{3/2})$ whose distances satisfy $d_{X} =\Theta(\sqrt{N}),\; d_{Z} = \Theta(\sqrt{N})$. The natural question is whether this code admits efficient decoding against data and syndrome errors of weight $\Theta(\sqrt{N})$. 




\subsubsection{Construction II}
\begin{theorem}\label{thm:expandernoisyrecovery}
    Let $H$ be a parity-check matrix with the corresponding code $C$ having parameters $[n,k,d]$. If $C$ has a syndrome decoding algorithm that runs in time $\Theta(n)$ then there exists an $H'$ such that $HGP(H',H')$ has a noisy syndrome exact recovery algorithm that runs in time $\Theta(N^{3/2})$.
\end{theorem}
\begin{proof}
Let $H$ be a parity-check matrix for the code, $C$ and let $[I_{n-k}:P_{(n-k)\times k}]$ be the systematic form for $H$. Let $[I_{k'} :P_{k' \times (k-k')}']$ be the systematic form for the matrix $P$ for some $k' := rank(P)$. 

Now consider the $(2k \times 2k)$ matrix $$H':=\begin{bmatrix}
    I_{k-k'} & P''\\
    P''^{(t)} & I_{k-k'}
\end{bmatrix},$$ where $P_{(k-k') \times (k+k')}'':=\begin{bmatrix} P'^{(t)}_{(k-k') \times k'} & I_{(k-k')}  &  P'^{(t)}_{(k-k') \times k'}\end{bmatrix}$ and $P'^{(t)}$ and $P''^{(t)}$ are the transpose matrices of $P'$ and $P''$ respectively. 

We show that the HGP code,  $HGP(H',H')$ is noisy syndrome exact recoverable. First, let us compute the parameters of the $HGP(H',H')$ assuming $rank(P) < k$.

Consider $C'$ to be the code defined by the generator matrix $[I:P'^{(t)}]$ and let $\pi(C')=\{(c',c'):c'\in C'\}$. By Claim \ref{clm:symmetrization}, we have $Ker(H')=\pi(C')\cap\pi(C')^{\perp}=\pi(C')$.
This also implies that the $Rowspan([I:P''])\subseteq Rowspan([P''^{(t)}:I])$. This implies, by performing elementary row operations and column swaps (or, equivalently, relabelling of variables) and by Claim \ref{claim:self-orthogonal}, the matrix $H'$ can be rewritten as $\begin{bmatrix}
        P & O\\
        O & P\\
        A' & B'
    \end{bmatrix}$.
Finally, invoking the Rank-Nullity Theorem and Claim \ref{systematic-rank}, we obtain $dim(Ker(H'))=n-k-rank(P)$. For the distance of the code $\pi(C')$, observe that $C'=Ker([P':I])$ and as $[I:P']$ is the systematic form for $P$ hence $d(C')>d$ and thus  $d(\pi(C'))\geq 2d$. 
    
Therefore, we get gives us an HGP code with parameters $[[\Theta(n^{2}),\Theta((n-k-k')^{2}),\geq 2d]]$. 

Now, we show the exact recovery algorithm given a noisy syndrome. 
The corresponding noisy syndrome equations for $HGP(H', H')$ for the $X$-errors (similarly for $Z$-errors) would be $\vs+\vE=(I\otimes H')\vx+(H'\otimes I)\vy,$ where $\vs \in \F_{q}^{4k^2}$ is the observed noisy syndrome with syndrome noise $\vE$ and 
$\begin{bmatrix}
    \vx\\
    \vy
\end{bmatrix}$ 
is the $X$ (or $Z$) data error.
Let us consider 
$E := \vE +(H'\otimes I)\vy$, and using the Kronecker product structure, 
we get 
This gives us,
\[\vs^{(i)}+E^{(i)}=H'\vx^{(i)} \qquad \text{ for } i \in [2k].\]
Applying an appropriate sequence of elementary row operations 
along with column swaps (which can also be visualized as relabelling of variables), combined with the Claims \ref{claim:self-orthogonal} and \ref{clm:error-weight}, one can obtain the following system of equations: $\vs'^{(i)}=H\begin{bmatrix}
    {E'}^{(i)}\\
    \vx^{(i)},
\end{bmatrix}$ where $\vs'^{(i)}$, and ${E'}^{(i)}$ are obtained from $\vs^{(i)}$ and ${E}^{(i)}$.

Now by performing syndrome decoding for $H$, we can recover ${E'}^{(i)}$ and as multiplication by elementary matrices are invertible,  we can recover $E^{(i)}$ and hence $E$. The running time is $\Theta(n^3)$ owing to the runtime of Gaussian elimination and matrix multiplication.
\end{proof}

\begin{corollary}\label{cor: ss-instatiation}
    Let $H=(L\cup R,E)$ be a $(c,d)-$biregular bipartite graph with $(\delta,c(1-\epsilon))$ vertex expansion with $\epsilon\in(0,1/4)$. Then $HGP(H',H')$ as defined in Theorem \ref{thm:expandernoisyrecovery}, has a $(\Theta(\sqrt{N}),\Theta(\sqrt{N}))-$  noisy syndrome exact recovery algorithm that runs in time $\Theta(N^{3/2})$.
\end{corollary}
\begin{proof}
    Follows from Theorem \ref{thm:expandernoisyrecovery}.
    We instantiate the specific choice of bipartite lossless expander in~\cite{golowich-lossless}, which allows \((\delta,c(1-\eps))\)-vertex expansion as in the assumption, with \(\delta=\Theta(\eps^5)\).
\end{proof}

\begin{corollary}\label{cor:zemor-syndrome-decoder}
Let $H$ be the parity-check matrix for the Tanner Code $C(G,C_{0})$, with $G$ being a d-regular $\lambda-$ spectral expander and $C_{0}$ being a $[d_{0},R_{0}d,\delta_{0}d]$ with $R_{0}> 1/2$ and $\delta_{0}> \lambda$. Then $HGP(H',H')$ as defined in Theorem \ref{thm:expandernoisyrecovery}, has a $(\Theta(\sqrt{N}),\Theta(\sqrt{N}))-$ noisy syndrome exact recovery algorithm that runs in time $\Theta(N^{3/2})$.
\begin{proof}
    Follows from Theorem \ref{thm:expandernoisyrecovery} and correctness of Algorithm \ref{algo:syndrome-zemor}.
\end{proof}
\end{corollary}
\section{Acknowledgements}
Vatsal thanks Prahladh Harsha for helpful conversations. The authors thank several anonymous reviewers for their insightful feedback, which  helped improve the presentation of the paper.
\bibliographystyle{alpha}
\bibliography{qlocal}

\appendix
\section{Appendix}

\subsection{Additional Preliminaries}\label{app:additional-prelims}

Here we cover the preliminaries that were omitted in the main body.

\subsubsection*{Classical error-correcting codes}

Let us acquaint ourselves with the essentials on classical codes.  Let \(\mb{F}_q\) be a finite field.  A \emph{(classical) code} $C$ of \emph{length} $n$ is any subset $C \subseteq \FF_q^n$, where we consider elements of the $n$-dimensional vector space \(\FF_q^n\) as row vectors of length $n$.  If $C \subseteq\FF_q^n$ is a linear subspace, then we call $C$ a \emph{linear} code.
The two most important parameters of a classical linear code are its \emph{dimension} $k(C) \coloneqq\dim_{\FF_q} C$ and its \emph{(minimum) distance}

\[d(C) \coloneqq \min\{d_H(x,y) : x,y\in C, x\neq y\} = \min\{\wgt_\Hamm(x) : x \in C,x\neq 0\},\]
where $d_H(x,y) \coloneqq |\{ i\in[n] : x_{i}\neq y_{i}\}|$
is the Hamming distance between vectors in $\FF_q^n$, and \(\wgt_\Hamm(x)\coloneqq|\{i\in[n]:x_i\ne0\}|\) is the Hamming weight of a vector $x\in\mathbb{F}_{q}^{n}$.
By a classical \([n,k,d]_q\) linear code, we mean a linear code having length \(n\), dimension \(k\), and distance \(d\), over the finite field \(\FF_q\).

There are two standard ways to present a classical linear code: with a \emph{generator matrix} $G$, or with a \emph{parity check matrix} $H$.
The former is any full-rank \(\FF_q\)-matrix whose row space is equal to $C$,
while the latter is any full-rank \(\FF_q\)-matrix whose kernel is equal to $C$.  More precisely, for a classical linear \([n,k,d]_q\) code, a generator matrix \(G\) is a \(k\times n\) matrix, and a parity check matrix \(H\) is a \((n-k)\times n\) matrix  characterized by the conditions
\[ C =\{xG\in\FF_q^n : x\in\FF_q^k\}= \{ c \in \FF_q^n : Hc^\transpose = 0\}.\]
By Gaussian elimination, we can obtain a parity check matrix from a generator matrix, and \emph{vice versa}, in polynomial time. A natural way of specifying a parity check matrix $H$ is via the systematic form where $H=\begin{bmatrix}I_{n-k}&P_{(n-k)\times k}\end{bmatrix}$, and then $G=\begin{bmatrix}P^\transpose_{k\times (n-k)}&I_{k}\end{bmatrix}$ will be a generator matrix for $C$.

Given a classical linear code $C$ of length $n$, the \emph{dual code} $C^\perp$ is another length $n$ code defined by
\[C^\perp \coloneqq \{y\in\FF_q^n: yc^\transpose=0\text{ for all } c\in C \}.\]
It follows by definitions that if $H$ is a parity check matrix and $G$ is a generator matrix for a linear code $C$, then $HG^\transpose=0$.  In other words, every parity check (resp. generator) matrix for \(C\) is a generator (resp. parity check) matrix for \(C^\perp\).

\subsubsection*{Quantum error-correcting codes}

We cover the essentials of quantum codes along with the stabilizer formalism for quantum codes introduced in \cite{Gottesman1997stabilizer} and \cite{@calderbank1998gf(4)}.

\paragraph*{Binary setting for quantum codes}

In the classical setting (for linear codes), codewords are composed of \emph{bits} that have \emph{values} in a \emph{field}.  In the quantum setting, codewords are composed of \emph{qubits} that have \emph{states} in a \emph{Hilbert space}.  The single qubit state space is the two-dimensional Hilbert space $\mathbb{C}^{(2)}$, and the $n$-qubit state space is the $n$-fold tensor product $(\mathbb{C}^{(2)})^{\otimes n}$.  For a single qubit, we have the fundamental states
\[
\ket{0}\coloneqq\begin{bmatrix}
    1\\0
\end{bmatrix}\quad\tx{and}\quad\ket{1}\coloneqq\begin{bmatrix}
    0\\1
\end{bmatrix},
\]
and every other state can then be represented by
\[
\ket{\psi}=\alpha\ket{0}+\beta\ket{1},\quad\tx{where }|\alpha|^2+|\beta|^2=1.
\]
For an \(n\)-qubit, by the tensor product structure of the state space, we have the fundamental states
\[
\ket{x}\coloneqq\ket{x_1}\otimes\cdots\otimes\ket{x_n},\quad\tx{ for each }x=(x_1,\ldots,x_n)\in\{0,1\}^n,
\]
and every other state can then be represented by
\[
\ket{\psi}=\sum_{x\in\{0,1\}^n}\lambda_x\ket{x},\quad\tx{where }\sum_{x\in\{0,1\}^n}|\lambda_x|^2=1.
\]

In general, a {\em quantum code of length $n$} is any $\CC$-linear subspace $\calC$ of the Hilbert space of $n$ qubits $(\CC^2)^{\otimes n}$.
Similar to the classical setting, the \emph{(minimum) distance} $d_Q$ of a quantum error-correcting code is defined to be the minimum number of qubits where non-trivial errors must occur in order to effect a non-trivial logical error on the code-space.
We call a quantum code of length $n$, with $\dim_\CC \calC = K$ and distance $d$ a $((n,K,d))$ quantum code.
If $K=2^k$ happens to be a power of 2, then we call $k$ the number of \emph{logical qubits} in the code, and call the code a $[[n,k,d]]$ quantum code.

\paragraph*{Stabilizer codes.}  We will be interested in quantum codes with more structure.  The \emph{Pauli group} on \(n\)-qubits is defined to be the group \(\mc{P}_n\) of linear operators \((\mb{C}^2)^{\otimes n}\to(\mb{C}^2)^{\otimes n}\) generated by elements \(M_1\otimes\cdots\otimes M_n\), where each \(M_i\in\{I,X,Y,Z\}\) with
\[
I\coloneqq\begin{bmatrix}
    1&0\\0&1
\end{bmatrix},\quad X\coloneqq\begin{bmatrix}
    0&1\\1&0
\end{bmatrix},\quad Y\coloneqq\begin{bmatrix}
    0&-i\\i&0
\end{bmatrix},\quad\tx{and}\quad Z\coloneqq\begin{bmatrix}
    1&0\\0&-1
\end{bmatrix}.
\]
For any \(a,b\in\mb{F}_2^n\), denote \(X(a)Z(b)=\bigotimes_{i=1}^{(n)}X^{a_i}Z^{b_i}\).  Since we have the relation \(Y=iXZ\), it follows that we have
\begin{align*}
\mc{P}_n&=\{i^\lambda M_1\otimes\cdots\otimes M_n:\lambda\in\{0,1,2,3\},\,M_i\in\{I,X,Y,Z\}\tx{ for all }i\in[n]\}\\
&=\{i^\lambda X(a)Z(b):\lambda\in\{0,1,2,3\},\,a,b\in\mb{F}_2^n\}.
\end{align*}
 A \emph{stabilizer subgroup} is defined to be any Abelian subgroup $\mathcal{S}$, of $\mathcal{P}_{n}$ not containing $-I$.  If \(\mc{S}\) is a stabilizer subgroup of $\mathcal{P}_{n}$, the \emph{stabilizer code} $\mc{C}(\mathcal{S})$ associated to it is defined to be the joint $(+1)$-eigenspace of the operators in $S$, that is,
\[
\mc{C}(\mathcal{S})=\{\ket{\psi}:\; g\ket{\psi}=\ket{\psi} \text{ for all }g\in\mc{S}\}.
\]
Moreover, the stabilizer code $\mc{C}(\mc{S})$ is said to encode \emph{$k$-logical qubits} if $dim(\mc{C}(\mathcal{S}))=2^{k}$.  It follows from basic group theory that if \(\mc{S}\) has order $2^{k}$, then the code \(\mc{C}(\mc{S})\) has dimension $2^{n-k}$.

\paragraph*{CSS codes.}  The seminal works \cite{calderbankshor98} and \cite{Steane96} showed how to build certain quantum error-correcting codes---now called \emph{CSS codes}--- using any pair of classical binary linear codes $(C_1,C_2)$ with $C_2^\perp \subseteq C_1$.
If $C_1$ and $C_2$ have parameters $[n,k_1,d_1]$ and $[n,k_2,d_2]$, then the CSS code $CSS(C_1, C_2)$ has parameters $[[n, k_1+k_2-n, d_Q]]$ where
\[ d_Q \coloneqq \min\{\wgt_\Hamm(a):a\in (C_1\setminus C_2^{\perp})\cup (C_2\setminus C_1^\perp )\} \]
is the (quantum) distance of $CSS(C_1, C_2)$.
If the parity check matrices of $C_1$ and $C_2$ are $H_1$ and $H_2$, then the condition that $C_2^\perp \subseteq C_1$ is equivalent to the condition $H_1 H_2^\transpose=O$, and we will write $CSS(H_1,H_2)$ to mean $CSS(C_1,C_2)$, and may refer to $H_1$ and $H_2$ as the \emph{quantum parity check matrices} of the CSS code.

\paragraph*{Non-binary setting for quantum codes.}

In this case, we consider \emph{qudits} instead of qubits.  The single qudit state space is the \(q\)-dimensional Hilbert space $\mathbb{C}^q$, and the $n$-qudit state space is the $n$-fold tensor product $(\mathbb{C}^q)^{\otimes n}$.  For a single qudit, we have the fundamental states
\[
\ket{0}\coloneqq\begin{bmatrix}
    1\\0\\\vdots\\0
\end{bmatrix},\quad\ket{1}\coloneqq\begin{bmatrix}
    0\\1\\\vdots\\0
\end{bmatrix},\quad\ldots,\quad\ket{q-1}\coloneqq\begin{bmatrix}
    0\\0\\\vdots\\1
\end{bmatrix}.
\]
and every other state can then be represented by
\[
\ket{\psi}=\sum_{i\in[0,q-1]}\alpha_i\ket{i},\quad\tx{where }\sum_{i\in[0,q-1]}|\alpha_i|^2=1.
\]
The further notions are simple-minded generalizations of what we saw in the \(q=2\) case, and we do not dwell on them here.

\subsection{Proofs of some  Helpful Claims}\label{app:helpful}

\begin{proof}[Proof of Claim \ref{claim:self-orthogonal}]
 The code $\pi(C)$ having parameters $[2n,k,2d]$ is direct and follows from the observation that the map $\pi:\mathbb{F}_{2}^{n}\rightarrow\mathbb{F}_{2}^{2n}$
 is an injective map. Now, to prove that $\pi(C)$ is self-orthogonal, we use the observation that \[<\pi(c),\pi(c')>=<(c,c),(c',c')>=<c,c'>+<c,c'>=0,\]
 as we are performing arithmetic over $\mathbb{F}_{2}$. 

 Now, for the parity-check matrix for the code, $\pi(C),$ observe that \[G':=[G:G],\] is a generator matrix for $\pi(C)$, where $G$ being a generator matrix for $C$. What one could now observe is the following: 
 \[[G:G]\begin{bmatrix}
     H^{T}\\
     O_{n\times (n-k)}
 \end{bmatrix}=O_{k\times (n-k)}.\]
 and similarly for $[O_{(n-k)\times n}:H]$. Now, using the fact:
 \[dim(\pi(C)^{\perp})=2n-dim(\pi(C))=2n-k,\]
 and the observation that the rows of $[H:O_{(n-k)\times n}]$ and $[O_{(n-k)\times n}:H]$ are linearly independent with one another, hence $rank([A:B])=(2n-k)-((n-k)+(n-k))=k.$
\end{proof}

\begin{proof}[Proof of Claim \ref{clm:symmetrization}]
    Consider a vector $x\in\mathbb{F}_{2}^{n}$ such that:
    \[H'x=O_{n\times 1}.\]

    The above equation can be further broken down into :
    \[[I:P]x=O_{(n-k)\times 1}\]
    and,
    \[[P^{(t)}:I]x=O_{k\times 1}.\]
    This implies $x\in C'$ and $x\in {C^{'}}^{\perp}$ or equivalently $x\in C'\cap {C'}^{\perp}$. 
\end{proof}

\begin{proof}[Proof of Claim \ref{systematic-rank}]
    Let $rank(P) < k$, and for the sake of contradiction, let us assume that $d(C')<d$. 
    As $rank(P)<k$, there exists a $x\in C':=Ker(P)$ such that $x\neq 0$ and $\wgt_\Hamm(x)=d(C')$. 
    Now, consider the vector $y:=(0_{k\times 1},x)$. It can be observed that $y\in C:=Ker(H)$. This implies $d(C)<d$, giving a contradiction.
\end{proof}
Let $C$ be a $[n,k,d]$ code with the systematic form of the parity-check matrix being $[I_{n-k}:P_{(n-k)\times k}]$. Now let us consider the matrix, \[
		H'=\begin{bmatrix}
			O & P\\
			P^\transpose & O
		\end{bmatrix},
		\]then:

\begin{proof}[Proof of Claim \ref{clm:distance-lem}]
    To find the minimum distance of the code  corresponding to the parity check matrix $H'$, we consider the solution of the system of equations:
    $$H'x=O_{n\times 1},$$
    where $x\in\mathbb{F}^{(n)}$.  The above system of equation is equivalent to solving the following set of equations
    \begin{align*}
        Px'=O_{(n-k)\times 1}\quad\tx{and}\quad P^\transpose x''=O_{k\times 1},\quad\tx{where }x'=\begin{bmatrix}
        x_{n-k+1}\\
       \vdots\\x_{n}\end{bmatrix},\,x''=\begin{bmatrix}
        x_{1}\\
       \vdots\\
        x_{n-k}
    \end{bmatrix}.
    \end{align*}
    Now consider the following cases for the column ranks of $P$ and $P^\transpose$:
    \begin{itemize}
        \item $\rank(P)<k$: This implies that the system $Px'=O_{(n-k)\times 1}$ has a non-zero solution, say $\alpha$ $\in\mathbb{F}_{q}^k$  and assume it is of the least Hamming weight. It is evident that $\alpha$ can be extended to a non-zero solution of $Hx=O$ with Hamming weight $\wgt_\Hamm(\alpha)$. This gives the lower bound $\wgt_\Hamm(\alpha)\geq d(C)\geq \min\{d(C),\; d(C^{\perp})\}$.
        \item $\rank(P^\transpose)<n-k$: This case can be handled similarly as the previous case and will give a lower bound of $d(C^{\perp})\geq \min\{d(C),d(C^{\perp})\}$ on the minimum distance of the code given by parity check matrix of $H'$.
        \item $\rank(P)=k$ and $\rank(P^\transpose)=n-k$: 
        If $\rank(P)=k$ and $\rank(P^\transpose)=n-k$ then using the fact that column-rank of a matrix is equal to its row-rank, we get that:
        $n-k\geq k$, corresponding to $\rank(P)=k$ and likewise for $\rank(P^\transpose)=n-k$ we get $k\geq n-k$. This implies $k=n/2$, and the claim is then easy.\qedhere
    \end{itemize}

\begin{proof}[Proof of Claim \ref{clm:error-weight}]
Let $H\in\{0,1\}^{m\times n}$ such that:
\[H:=\begin{bmatrix}
    H_{1,1} & H_{1,2} & . &. & H_{1,n}\\
    & & .& & &\\
    & & .& & &\\
    & & .& & &\\
    H_{m,1} & H_{m,2} &. &. & H_{m,n}
\end{bmatrix}\]
By the Kronecker product representation of tensor product, we have:
    \[H\otimes I=\begin{bmatrix}
        H_{1,1}\cdot I_{n\times n} &. & .& . & H_{1,n}\cdot I_{n\times n}\\
        & &\emph{.} & &\\
        &  &\emph{.} & &\\
        H_{m,1}\cdot I_{n\times n} & . &. & . &H_{m,n}\cdot I_{n\times n} 
    \end{bmatrix}\]
    This implies,
    \[\wgt_\Hamm(E^{i})=\wgt_\Hamm(H_{i,1}\cdot x^{1}+...+H_{i,n}\cdot x^{n}),\]
    where $x=(x^{1},...,x^{n})$ with $x^{l}\in\mathbb{F}_{2}^{n}$ for $1\leq l\leq n$.

    Now by triangle inequality, we obtain:
    \[\wgt_\Hamm(E^{i})=\wgt_\Hamm(H_{i,1}\cdot x^{1}+...+H_{i,n}\cdot x^{n})\leq \wgt_\Hamm(x^{1})+...+\wgt_\Hamm(x^{n})= \wgt_\Hamm(x),\]
    and the second-last inequality follows from the fact that $H_{i,j}\in\{0,1\}$
\end{proof}
\end{proof}

We now reduce the noisy syndrome exact recovery problem to that of usual syndrome decoding in the absence of noise. We use the following  observation repeatedly:

\begin{proof}[Proof of Claim \ref{noiseless syndrome reduction}]
    The syndrome equation for $E$ can be rewritten as:
    \begin{equation}
\begin{aligned}
s^{\mathrm{obs}}
&= [\, I_{2k \times 2k} : H' \,]
\begin{bmatrix}
E\\
a
\end{bmatrix} \\
&=
\begin{bmatrix}
I & O & : &P & O\\
O & I & : & O & P
\end{bmatrix}
\begin{bmatrix}
E\\
a
\end{bmatrix}
\end{aligned}
\end{equation}

    Based on the structure of the coefficient matrix, we obtain the system of equations:
    \[s^{obs}_{1}=[I:P]\begin{bmatrix}
        E^{1}\\
        a^{1}
    \end{bmatrix}\] and
 \[s^{obs}_{1}=[I:P]\begin{bmatrix}
        E^{2}\\
        a^{2}
    \end{bmatrix}\] 

where
    $E^{i}=\begin{bmatrix}
        E_{k\cdot (i-1)+1}\\
        E_{k\cdot (i-1)+1}\\
        .\\
        .\\
        E_{k\cdot i}\\
        \end{bmatrix}
       $
       and,
        $a^{i}=\begin{bmatrix}
        a_{k\cdot (i-1)+1}\\
        .\\
        .\\
        a_{k\cdot i}
    \end{bmatrix}$
    where $i\in\{1,2\}$.

    The two systems of equations are nothing but two instances of (noiseless) syndrome decoding. 
\end{proof}
\begin{remark}
    The above proof also works (upto minor modifications) for the case when $H:=\begin{bmatrix}
        O & P\\
        P^{t} & O
    \end{bmatrix}$, where we obtain syndrome decoding equations corresponding to the code and its dual.
\end{remark}

\subsection{Review of Sipser-Spielman Decoding}

    In this section we will give a variant of the Sipser-Spielman \cite{SipserSpielman1996expander} decoder for bipartite expander codes. 
     We first review the bit-flip decoding algorithm provided by Sipser-Spielman codes and then provide a variant that is useful for decoding HGP codes obtained from bipartite vertex expanders. The variant is based on the fact that while for classical codes the decoding algorithm
     has access to both the received vector and the syndrome, in the quantum case we only have access to the syndrome bits. 

    Here we look at the bipartite graph defined by the parity-check matrix of the code, where the left partition corresponds to the variable nodes and the right partition corresponds to the check nodes. There is an edge between a variable node and a check node if the variable appears in that particular row of that parity-check matrix. In Algorithm \ref{alg:expander-decoding}, the received vector is placed on the variable nodes, and the syndrome vector is placed on the check nodes. In Algorithm \ref{alg:expander-decoding}, the idea is that for every variable bit with more unsatisfied check bits as its neighbors than the satisfied check bits, we flip the value of  that particular variable bit and update the corresponding check bits. Algorithm \ref{alg:expander-decoding-flags} pursues a similar idea. Here instead of having access to both the received vector and the syndrome vector, the decoder has access to only the syndrome vector. Also, for every variable node we associate a flag variable. The flag variables are a placeholder for the received vector and are initialized to the all zero vector. The syndrome variant proceeds in a similar fashion like the previous algorithm. The idea is that for every variable node with more unsatisfied check bits as its neighbors than the satisfied check bits, we flip the value of the associated flag variable and update the corresponding check bits. After the termination of the Algorithm \ref{alg:expander-decoding-flags}, we claim that the vector in the flag variables contains the $X$(or $Z$-) error vector -- see the formal proof in Claim \ref{clm: ss-syndrome-decoding}.

\begin{algorithm}[H]
\caption{Sipser-Spielman Decoding Algorithm \cite{SipserSpielman1996expander}}
\label{alg:expander-decoding}
\begin{algorithmic}[1]
\Require A $(c,d)$-regular graph $B$ between variables and constraints, and an assignment of values to the variables.
\Ensure A codeword (values of the variables), or \textsc{failed to decode}.

\vspace{0.5em}
\State \textbf{Set-up phase:}
\State For each constraint, determine whether or not it is satisfied by the variables.
\State Initialize sets $S_0, \ldots, S_c$ to empty sets.
\For{each variable $v$}
    \State Count the number $i$ of unsatisfied constraints in which $v$ appears.
    \State Place $v$ in set $S_i$.
\EndFor

\vspace{0.5em}
\State \textbf{Loop:}
\While{$S_{\lceil c/2 \rceil}, \ldots, S_c$ are not all empty}
    \State Find the greatest $i$ such that $S_i$ is not empty.
    \State Choose a variable $v$ from set $S_i$.
    \State Flip the value of variable $v$.
    \For{each constraint $C$ that contains variable $v$}
        \State Update the status of constraint $C$.
        \For{each variable $w$ in constraint $C$}
            \State Recompute the number of unsatisfied constraints in which $w$ appears.
            \State Move $w$ to the set indexed by this number.
        \EndFor
    \EndFor
\EndWhile

\vspace{0.5em}
\If{all variables are in $S_0$}
    \State \textbf{output} the values of the variables.
\Else
    \State \textbf{report} \textsc{failed to decode}.
\EndIf
\end{algorithmic}
\end{algorithm}

\begin{algorithm}[H]
\caption{Syndrome Decoding Sipser-Spielman variant}
\label{alg:expander-decoding-flags}
\begin{algorithmic}[1]
\Require A $(c,d)$-regular graph $B$ between variables $\{v_1,\ldots,v_n\}$ and constraints, and an assignment of values to the variables.
\Ensure A codeword (values of the variables) together with flip flags $f_1,\ldots,f_n$, or \textsc{failed to decode}.

\vspace{0.5em}
\State \textbf{Set-up phase:}
\State For each constraint, determine whether or not it is satisfied by the variables.
\State Initialize sets $S_0, \ldots, S_c$ to empty sets.
\For{each variable $v_j$, $j = 1, \ldots, n$}
    \State $f_j \gets 0$ \Comment{flag = 1 iff $v_j$ has been flipped}
    \State Count the number $i$ of unsatisfied constraints in which $v_j$ appears.
    \State Place $v_j$ in set $S_i$.
\EndFor

\vspace{0.5em}
\State \textbf{Loop:}
\While{$S_{\lceil c/2 \rceil}, \ldots, S_c$ are not all empty}
    \State Find the greatest $i$ such that $S_i$ is not empty.
    \State Choose a variable $v_j$ from set $S_i$.
    \State Flip the value of variable $v_j$.
    \State $f_j \gets 1$ \Comment{mark $v_j$ as flipped}
    \For{each constraint $C$ that contains variable $v_j$}
        \State Update the status of constraint $C$.
        \For{each variable $w$ in constraint $C$}
            \State Recompute the number of unsatisfied constraints in which $w$ appears.
            \State Move $w$ to the set indexed by this number.
        \EndFor
    \EndFor
\EndWhile

\vspace{0.5em}
\If{all variables are in $S_0$}
    \State \textbf{output} flip flags $f_1,\ldots,f_n$.
\Else
    \State \textbf{report} \textsc{failed to decode}.
\EndIf
\end{algorithmic}
\end{algorithm}
\newpage
In Figure \ref{syndrome-SS} and Figure \ref{syndrome-SS2}, we illustrate the working of Algorithm \ref{alg:expander-decoding-flags} for some error $x$. We demonstrate how the flag variables are updated when a variable bit is flipped.
\begin{figure}[H]
\centering

\begin{tikzpicture}[
    vnode/.style={circle, draw=blue!80, fill=blue!5, thick, minimum size=8mm},
    cnode/.style={rectangle, draw=black!80, fill=gray!10, thick, minimum size=7mm},
    flagged/.style={circle, draw=orange!90, fill=orange!20, ultra thick, minimum size=8mm},
    unsat/.style={rectangle, draw=red!80, fill=red!10, thick, minimum size=7mm},
    node distance=1.2cm
]

\node[vnode] (V1) at (0, 0)  {$x_1$}; 
\node[left=2pt of V1, font=\footnotesize] {$f_1=0$};

\node[flagged] (V2) at (0, -2) {$x_2$}; 
\node[left=2pt of V2, font=\footnotesize, orange!80!black] {\textbf{$f_2=0$}};

\node[vnode] (V3) at (0, -4) {$x_3$}; 
\node[left=2pt of V3, font=\footnotesize] {$f_3=0$};

\node[unsat] (C1) at (4, -0.5) {1}; 
\node[cnode] (C2) at (4, -2) {0}; 
\node[unsat] (C3) at (4, -3.5) {1};

\draw[thick, orange] (V2) -- (C1); 
\draw[thick] (V2) -- (C2);
\draw[thick, orange] (V2) -- (C3);

\draw (V1) -- (C1); \draw (V1) -- (C2);
\draw (V3) -- (C2); \draw (V3) -- (C3);

\node[above=0.5cm of V1, font=\bfseries] {Variables};
\node[above=0.5cm of C1, font=\bfseries] {Checks};

\end{tikzpicture}
\caption{Initialized  state for Algorithm \ref{alg:expander-decoding-flags} where each flag variable is set to $0$.}
\label{syndrome-SS}
\end{figure}

\begin{figure}[H]
\centering
\begin{tikzpicture}[
    vnode/.style={circle, draw=blue!80, fill=blue!5, thick, minimum size=8mm},
    cnode/.style={rectangle, draw=black!80, fill=gray!10, thick, minimum size=7mm},
    flagged/.style={circle, draw=orange!90, fill=orange!20, ultra thick, minimum size=8mm},
    unsat/.style={rectangle, draw=red!80, fill=red!10, thick, minimum size=7mm},
    node distance=1.2cm
]

\node[vnode] (V1) at (0, 0)  {$x_1$}; 
\node[left=2pt of V1, font=\footnotesize] {$f_1=0$};

\node[flagged] (V2) at (0, -2) {$\bar{x}_{2}$}; 
\node[left=2pt of V2, font=\footnotesize, orange!80!black] {\textbf{$f_2=1$}};

\node[vnode] (V3) at (0, -4) {$x_3$}; 
\node[left=2pt of V3, font=\footnotesize] {$f_3=0$};

\node[unsat] (C1) at (4, -0.5) {0}; 
\node[cnode] (C2) at (4, -2) {1}; 
\node[unsat] (C3) at (4, -3.5) {0};

\draw[thick, orange] (V2) -- (C1); 
\draw[thick] (V2) -- (C2);
\draw[thick, orange] (V2) -- (C3);

\draw (V1) -- (C1); \draw (V1) -- (C2);
\draw (V3) -- (C2); \draw (V3) -- (C3);

\node[above=0.5cm of V1, font=\bfseries] {Variables};
\node[above=0.5cm of C1, font=\bfseries] {Checks};

\end{tikzpicture}
\caption{After flipping second bit}
\label{syndrome-SS2}
\end{figure}

\begin{claim}\label{clm: ss-syndrome-decoding}
    Let $C_{G}$ be an expander code defined with respect to a $(c,d)$ biregular $(\delta,c(1-\epsilon))$ bipartite vertex expander, $G$, with $\epsilon\in(0,1/4)$. Then every error $e\in\mathbb{F}_{2}^{n}$, $\wgt_\Hamm(e)\leq \delta(1-2\epsilon)n $ can be corrected in linear time by Algorithm \ref{alg:expander-decoding-flags}.
\end{claim}
\begin{proof}
    By the correctness of the Sipser-Spielman bit-flip decoding Algorithm \ref{alg:expander-decoding} as mentioned in \cite{SipserSpielman1996expander}, the flag vector contains the error vector, e, if $\wgt_\Hamm(e)\leq \delta (1-2\epsilon)n$.
\end{proof}

\subsection{Syndrome version of Z\'emor's Decoding Algorithm}
In this subsection we mention a variant of the Z\'emor's Decoding algorithm that has access to only the syndrome vector as in Algorithm \ref{algo:syndrome-zemor}. We start by first recalling the original Z\'emor's decoding algorithm \ref{algo:zemor} as mentioned in \cite{zemor2002expander}, which has access to the received word and the syndrome vector. We then provide Algorithm \ref{algo:syndrome-zemor} where we only have access to the syndrome vector. To the best of our knowledge, the syndrome-based version was not known before. In Algorithm \ref{algo:zemor}, the left and right partitions of the bipartite graph are of size $n$ with each node having a constant degree $d$ and the received vector lying on the edges. We also have an inner code $C_{0}$ of length $d$ associated with the decoder. For every node, the decoder considers the edges incident on the node and the corresponding bits of the received vector lying on these edges. For any node $v$, we call these bits the ``local picture'' of the node $v$. If the local picture of node $v$ is not a codeword in the inner code $C_{0}$ then it is first decoded to the nearest codeword (in Hamming metric) in $C_{0}$. In Algorithm \ref{algo:syndrome-zemor}, the only thing available is the syndrome vector. As the syndrome equations are a system of linear equations over $\mathbb{F}_{2}$ and with the original error vector as a solution, the system is feasible. Hence using any linear system solver over $\mathbb{F}_{2}$, the decoder first finds a vector, say $r$, that satisfies the syndrome equations. Having done this, the decoder then makes a call to Algorithm \ref{algo:zemor} with $r$ as the received vector. 

It is important to note that given the parity-check matrix $H$ for the inner code $C_{0},$ we can construct $\mathbf{H}_{global}$ by using the fact that every vertex the set of constraints are the ones induced by $H$. We illustrate this in the figure below.  

\begin{tikzpicture}[scale=0.8]

\node[circle, draw, thick, fill=blue!15, minimum size=22pt, font=\normalsize] (v) at (0, 0) {$v$};

\foreach \i/\lbl in {1/u_1, 2/u_2, 3/u_3} {
  \node[circle, draw, thick, fill=red!15, minimum size=22pt, font=\small]
    (u\i) at (3.0, 2.2 - \i*1.3) {$\lbl$};
}

\node[font=\large] at (3.0, -2.7) {$\vdots$};

\node[circle, draw, thick, fill=red!15, minimum size=22pt, font=\small]
  (ud) at (3.0, -3.8) {$u_d$};

\draw[thick] (v) -- node[above, font=\small] {$e_1$} (u1);
\draw[thick] (v) -- node[above, font=\small] {$e_2$} (u2);
\draw[thick] (v) -- node[above, font=\small] {$e_3$} (u3);
\draw[thick] (v) -- node[below, font=\small] {$e_d$} (ud);

\draw[->, thick, gray] (3.8, -0.8) -- (5.5, -0.8);
\node[font=\small, gray, above] at (4.6, -0.8) {syndrome};
\node[font=\small, gray, below] at (4.6, -0.8) {constraint};

\draw[thick, fill=blue!8] (7.5, 0.8) rectangle (10.5, -2.4);
\node[font=\normalsize] at (9.0, -0.8) {$H$};

\draw[decorate, decoration={brace, amplitude=5pt}, thick]
  (7.5, 1.1) -- (10.5, 1.1)
  node[midway, above=6pt, font=\small] {$d$ columns};

\draw[decorate, decoration={brace, amplitude=5pt, mirror}, thick]
  (7.2, 0.8) -- (7.2, -2.4)
  node[midway, left=6pt, font=\small, align=center] {$r$ rows};

\draw[thick] (11.2, 0.8) rectangle (11.8, -2.4);
\node[font=\small, rotate=0] at (11.5, -0.8) {$x|_v$};

\draw[decorate, decoration={brace, amplitude=4pt}, thick]
  (12.0, 0.8) -- (12.0, -2.4)
  node[midway, right=5.5pt, font=\small, align=left,rotate=270] {$d \times 1$};

\node[font=\large] at (12.8, -0.6) {$=$};

\draw[thick, fill=yellow!15] (13.5, 0.8) rectangle (14.1, -2.4);
\node[font=\small, rotate=00] at (13.8, -0.8) {$s_v$};

\draw[decorate, decoration={brace, amplitude=4pt}, thick]
  (14.4, 0.8) -- (14.4, -2.4)
  node[midway, right=5pt, font=\small, align=left, rotate=270] {$r \times 1$};

\node[font=\small, align=center, text width=5cm] at (9.0, -3.8) {%
  $x|_v = \bigl(x_{e_1},\; x_{e_2},\; \ldots,\; x_{e_d}\bigr)^\top$\\[4pt]
  {\footnotesize (bits on the $d$ edges incident to $v$)}
};

\end{tikzpicture}

\begin{minipage}{0.85\textwidth}
\captionof{figure}{Local view of $\mathbf{H}_{\mathrm{global}}$ at vertex $v$. The row block for $v$ (highlighted) has only $d$ nonzero columns, placed at the edges $E_v = \{e_1, \ldots, e_d\}$ incident to $v$, reducing to the local syndrome equation $H \cdot x|_v = s_v$.}
\label{fig:local-view}
\end{minipage}
\begin{algorithm}[H]
\caption{Z\'emor's Decoding Algorithm}
\label{algo:zemor}
\begin{algorithmic}[1]
\Require Bipartite $d$-regular expander $G = (L \cup R, E)$ with $|L| = |R| = n$;
    component code $C_0$, a $[d, r_0 d, \delta_0]$ code;
    received word $x \in \mathbb{F}_2^{|E|}$
\Ensure Decoded codeword $x \in \mathcal{C}$, or \textsc{Failure}
\Statex \Comment{$x|_v$ denotes the $d$ bits of $x$ on edges incident to vertex $v$}
\State $V^{(0)} \gets L$;\quad $V^{(1)} \gets R$
\State $i \gets 0$
\While{$\exists\, v \in L \cup R$ such that $x|_v \notin C_0$}
    \For{each $v \in V^{(i \bmod 2)}$ with $x|_v \notin C_0$}
        \State Decode $x|_v$ to its nearest codeword in $C_0$
        \State Update $x|_e$ for all $e \in E_v$ accordingly
    \EndFor
    \State $i \gets i + 1$ \Comment{alternate sides: $L \to R \to L \to \cdots$}
\EndWhile
\If{$x|_v \in C_0$ for all $v \in L \cup R$}
    \State \Return $x$
\Else
    \State \Return \textsc{Failure}
\EndIf
\end{algorithmic}
\end{algorithm}

\bigskip

\begin{algorithm}[H]
\caption{Syndrome-Based Z\'emor's Decoding}
\label{algo:syndrome-zemor}
\begin{algorithmic}[1]
\Require Bipartite $d$-regular expander $G = (L \cup R, E)$ with $|L| = |R| = n$;
    inner code $C_0$, a $[d, r_0 d, \delta_0 d]$ code with parity-check matrix
    $H \in \mathbb{F}_2^{r \times d}$;
    syndrome vector $(s_v)_{v \in L \cup R}$ where $s_v \in \mathbb{F}_2^{r}$
\Ensure Error estimate $\hat{e} \in \mathbb{F}_2^{|E|}$, or \textsc{Failure}
\Statex
\Statex \textbf{Step 1:} Find any vector consistent with the syndromes.
\State Construct the global syndrome equation $\mathbf{H}_{\mathrm{global}}\, r = \mathbf{s}$,
    where $\mathbf{H}_{\mathrm{global}}$ has a copy of $H$ at each vertex and
    $\mathbf{s} = (s_v)_{v \in L \cup R}$
\State Using Gaussian elimination, find $r \in \mathbb{F}_2^{|E|}$ such that
    $\mathbf{H}_{\mathrm{global}}\, r = \mathbf{s}$
\If{no solution exists}
    \State \Return \textsc{Failure}
\EndIf
\Statex
\Statex \textbf{Step 2:} Run standard Z\'emor's decoding on $r$ as the received word.
\State $c \gets \textsc{Z\'emorDecode}(G, C_0, r)$
\If{$c = \textsc{Failure}$}
    \State \Return \textsc{Failure}
\EndIf
\State $\hat{e} \gets r \oplus c$ \Comment{equivalently $r - c$ over $\mathbb{F}_2$}
\State \Return $\hat{e}$
\end{algorithmic}
\end{algorithm}

\subsubsection{Proof of correctness for Algorithm \ref{algo:syndrome-zemor}}
\begin{proof}
For the syndrome vector, compute a $r\in\mathbb{F}_{2}^{|E|}$ such that all the syndrome equations are satisfied. This can be done in $O(n^3)$ time via a linear system solver. Now, assume the data error on the transmitted codeword was $e$ with $\Delta(e)\leq \delta_{0}^{2}/4$, where $\Delta(x)$ is the fractional Hamming weight. Now, as $\mathbf{H}_{\mathrm{global}}\cdot r=s=\mathbf{H}_{\mathrm{global}}\cdot e$ hence $r+C=e+C$ and hence $r-e\in C$ or $r=c+e$ for some codeword $c$, where $\mathbf{H}_{\mathrm{global}}$ is the parity-check matrix for the expander code defined on the spectral expander. Now, by the correctness of Zémor's Decoding Algorithm, we have that on decoding the received vector $r$, we obtain the codeword $c$ and hence $r-c=e$. 
\end{proof}

\end{document}